\documentclass[conference,letterpaper]{IEEEtran}

\usepackage[utf8]{inputenc} 
\usepackage[T1]{fontenc}
\usepackage{url}
\usepackage{ifthen}
\usepackage{cite}
\usepackage[cmex10]{amsmath}
\usepackage{xcolor}

\usepackage{amssymb, amsthm, amsfonts, mathtools}

\newtheorem{theorem}{Theorem}
\newtheorem{lemma}{Lemma}
\newtheorem{corollary}{Corollary}

\theoremstyle{plain}

\usepackage{graphicx}
\usepackage{subcaption}

\newcommand{\part}[2]{\frac{\partial #1}{\partial #2}}
\newcommand{\N}{\mathcal{N}}
\newcommand{\R}{\mathbb{R}}
\newcommand{\E}{\mathbb{E}}
\newcommand{\HS}{\mathrm{HS}}
\DeclareMathOperator{\Tr}{Tr}
\DeclareMathOperator{\Var}{Var}
\DeclareMathOperator{\Cov}{Cov}
\DeclareMathOperator{\grad}{grad}
\DeclareMathOperator{\Hess}{Hess}

\begin{document}
\title{Convexity of mutual information\\along the Ornstein-Uhlenbeck flow} 

\author{\IEEEauthorblockN{Andre Wibisono}
\IEEEauthorblockA{
College of Computing \\
Georgia Institute of Technology \\
Atlanta, GA 30313 \\
wibisono@gatech.edu}
\and
\IEEEauthorblockN{Varun Jog}
\IEEEauthorblockA{
Department of Electrical \& Computer Engineering \\
University of Wisconsin - Madison\\
Madison, WI 53706 \\
vjog@wisc.edu}
}

\maketitle

\begin{abstract}
We study the convexity of mutual information as a function of time along the flow of the Ornstein-Uhlenbeck process.
We prove that if the initial distribution is strongly log-concave, then mutual information is eventually convex, i.e., convex for all large time.
In particular, if the initial distribution is sufficiently strongly log-concave compared to the target Gaussian measure, then mutual information is always a convex function of time.
We also prove that if the initial distribution is either bounded or has finite fourth moment and Fisher information, then mutual information is eventually convex.
Finally, we provide counterexamples to show that mutual information can be nonconvex at small time.
\end{abstract}

\section{Introduction}

The Ornstein-Uhlenbeck (OU) process is the simplest stochastic process after Brownian motion,
and it is the most general stochastic process for which we know the exact solution, which is an exponential interpolation of Gaussian noise. The OU process plays an important role in theory and applications.
The OU process provides an interpolation between any distribution and a Gaussian along a constant covariance path.
 This property has been a vital tool to prove the optimality of Gaussian in various information theoretic inequalities, including the entropy power inequality~\cite{Blachman65,Rioul11,MB07}.

In applications, the OU process appears as the continuous-time limit of basic Gaussian autoregressive models.
It also appears as the approximate dynamics of stochastic algorithms after linearization around stationary points~\cite{MHB16}. The OU process is a model example for general stochastic processes, and serves as a testbed to examine and test conjectures for the general case. Indeed, our approach to the OU process in this paper is motivated by a desire to understand how various information-theoretic quantities such as entropy, mutual information, and Fisher information evolve in general stochastic processes. This builds up on our previous works \cite{WibisonoJL17, WJ18}, where we carried out a similar analysis for the heat flow, and complements  recent investigations along closely related lines \cite{GuoEtAl11, Che15, Tos15a,  ZhaEtal18}.

Both the heat diffusion and the OU processes are examples of Fokker-Planck processes. The Fokker-Planck process is the sampling analogue of the usual gradient flow for optimization;
indeed, we can view the Fokker-Planck process as the gradient flow of the relative entropy functional in the space of measures with the Wasserstein metric~\cite{JKO98,OV00,Vil03,Vil08}.
This interpretation provides valuable information on the behavior of relative entropy.
For example, if the target measure is log-concave, then relative entropy is decreasing in a convex manner along the Fokker-Planck process. Such a result also follows from the seminal work of Bakry and Emery \cite{BakEm85, BakEtal13}, where diffusion processes have been examined in exquisite detail.

The behavior of mutual information, on the other hand, is not as well understood.
By the data processing inequality one may note that mutual information is decreasing along the Fokker-Planck process, which means the first time derivative is negative.
Furthermore, we can derive identities relating information and estimation quantities~\cite{WibisonoJL17}, which generalize the I-MMSE relationship for the Gaussian channel~\cite{GuoEtAl05}.
At the level of second time derivative, however, the behavior of mutual information is more complicated. In~\cite{WJ18} we studied the basic case of the heat flow, or the Brownian motion, and we showed there are interesting non-convex behaviors of mutual information especially at initial time. 

In this paper we study the corresponding questions for the OU process. Notice that the OU process may be obtained by rescaling time and space in the heat flow, and one could hope that convexity results for the heat flow \cite{WJ18} carry over with little work. However, this does not appear to be the case since rescaling does not preserve signs of higher order derivatives---a fact also noted in \cite{Che15} with regards to the derivatives of entropy obtained in \cite{GuoEtAl11}.
For simplicity in this paper we treat the case when the target Gaussian measure has isotropic covariance, but our technique extends to the general case.
We show that the results for the heat flow qualitatively extend to the OU process, but now with a subtle interplay with the size of the covariance of the target Gaussian measure.

Our first main result states that if the initial distribution is strongly log-concave, then mutual information is eventually convex.
In particular, if the initial distribution is sufficiently strongly log-concave compared to the target Gaussian measure, then mutual information is always convex.
We also prove that if the initial distribution is either bounded, or has finite fourth moment and Fisher information, then mutual information is eventually convex, with a time threshold that depends on the target covariance.
We also provide counterexamples to show that mutual information can be nonconvex at some small time.
In particular, when the initial distribution is a mixture of point masses, we show that mutual information along the OU process initially starts at the discrete entropy, which is the same behavior as seen in the heat flow.
In the limit of infinite target covariance, in which case the OU process becomes the Brownian motion, most of our results recover the corresponding results for the heat flow from~\cite{WJ18}.

\section{Background and problem setup}

\subsection{The Ornstein-Uhlenbeck (OU) process}

Let $\nu = \N(0,\frac{1}{\alpha} I)$ be the Gaussian measure in $\R^n$ with mean $0$ and isotropic covariance $\frac{1}{\alpha} I$, for some $\alpha > 0$.
The {\em Ornstein-Uhlenbeck} (OU) process in $\R^n$ with target measure $\nu = \N(0,\frac{1}{\alpha}I)$ is the stochastic differential equation
\begin{align}\label{Eq:OU}
dX = -\alpha X \, dt + \sqrt{2} \, dW
\end{align}
where $X = (X_t)_{t \ge 0}$ is a stochastic process in $\R^n$ and $W = (W_t)_{t \ge 0}$ is the standard Brownian motion in $\R^n$ starting at $W_0 = 0$.
The OU process admits a closed-form solution in terms of It\^o integral:
$X_t = e^{-\alpha t} ( X_0 + \sqrt{2} \int_0^t e^{\alpha s} \, dW_s)$.
At each time $t > 0$, 
the solution $X_t$ has equality in distribution
\begin{align}\label{Eq:OUSoln}
X_t \stackrel{d}{=} e^{-\alpha t} X_0 + \sqrt{\frac{1-e^{-2\alpha t}}{\alpha}} Z
\end{align}
where $Z \sim \N(0,I)$ is independent of $X_0$.
As $t \to \infty$, $X_t \stackrel{d}{\to} \frac{1}{\alpha} Z \sim \N(0,\frac{1}{\alpha} I)$ exponentially fast, so indeed $\nu = \N(0,\frac{1}{\alpha} I)$ is the target stationary measure for the OU process.

In the space of probability measures, the OU process~\eqref{Eq:OU} corresponds to the following partial differential equation:
\begin{align}\label{Eq:OUCE}
\part{\rho}{t} = \alpha \nabla \cdot (\rho x) + \Delta \rho.
\end{align}
Here $\rho = \rho(x,t)$ is a probability density over space $x \in \R^n$ for each time $t \ge 0$,  $\nabla \cdot = (\part{}{x_1},\dots,\part{}{x_n})^\top$ is the divergence, and $\Delta = \sum_{i=1}^n \part{^2}{x_i^2}$ is the Laplacian operator.
Concretely, if the random variable $X_t \sim \rho_t$ evolves following the OU process~\eqref{Eq:OU}, then its probability density function $\rho_t(x) = \rho(x,t)$ evolves following the equation~\eqref{Eq:OUCE} in the space of measures.

From the solution~\eqref{Eq:OUSoln} in terms of random variables, we also have the following solution for~\eqref{Eq:OUCE} in the space of measures: 
\begin{align}\label{Eq:OUCESoln}
\rho_t(y) = \frac{1}{(2 \pi \tau_\alpha(t))^{\frac{n}{2}}} \int_{\R^n} \rho_0(x) e^{-\frac{\|y-e^{-\alpha t} x\|^2}{2\tau_\alpha(t)}} \, dx
\end{align}
where $\tau_\alpha(t) = \frac{1}{\alpha}(1-e^{-2\alpha t})$ for $t > 0$.
One may also directly verify that~\eqref{Eq:OUCESoln} satisfies the equation~\eqref{Eq:OUCE}.
We refer to the flow of~\eqref{Eq:OUCE}, or the exact solution~\eqref{Eq:OUCESoln} above, as the {\em Ornstein-Uhlenbeck (OU) flow} in the space of measures.
Note that $\tau_\alpha(t) \to 2t$ as $\alpha \to 0$, and the Ornstein-Uhlenbeck process above recovers the Brownian motion or the heat flow.

\subsection{Derivatives of relative entropy along the OU flow}

For a reference probability measure $\nu$ on $\R^n$, let 
$$H_\nu(\rho) = \int_{\R^n} \rho(x) \log \frac{\rho(x)}{\nu(x)} \, dx$$
denote the {\em relative entropy} of a probability measure $\rho$ with respect to $\nu$.
This is also known as the Kullback-Leibler (KL) divergence.
Relative entropy has the property that $H_\nu(\rho) \ge 0$, and $H_\nu(\rho) = 0$ if and only if $\rho = \nu$.

The {\em relative Fisher information} of $\rho$ with respect to $\nu$ is
$$J_\nu(\rho) = \int_{\R^n} \rho(x) \left\| \nabla \log \frac{\rho(x)}{\nu(x)} \right\|^2 \, dx.$$
The {\em relative second-order Fisher information} of $\rho$ with respect to $\nu$ is
$$K_\nu(\rho) = \int_{\R^n} \rho(x) \left\| \nabla^2 \log \frac{\rho(x)}{\nu(x)} \right\|_{\HS}^2 \, dx$$
where $\|M\|_{\HS}^2 = \sum_{i,j=1}^n M_{ij}^2$ is the Hilbert-Schmidt (or Frobenius) norm of a symmetric matrix $M = (M_{ij}) \in \R^{n \times n}$.
For $X \sim \rho$, we also write $H_\nu(X)$, $J_\nu(X)$, and $K_\nu(X)$ in place of $H_\nu(\rho)$, $J_\nu(\rho)$, and $K_\nu(\rho)$, respectively.

We recall the interpretation of the OU flow as the gradient flow of relative entropy with respect to $\nu = \N(0,\frac{1}{\alpha} I)$ in the space of probability measures with the Wasserstein metric; see for example~\cite{Vil03,Vil08,OV00},
or Appendix~\ref{App:FP}.
This allows us to translate general gradient flow relations to obtain the following identities for the time derivatives of relative entropy along the OU flow.

\begin{lemma}\label{Lem:OUDer}
Along the OU flow for $\nu = \N(0,\frac{1}{\alpha}I)$,
\begin{align*}
\frac{d}{dt} H_\nu(X_t) &= -J_\nu(X_t), \\
\frac{d^2}{dt^2} H_\nu(X_t) &= 2K_\nu(X_t) + 2\alpha J_\nu(X_t).
\end{align*}
\end{lemma}

Note that $J_\nu(X_t) \ge 0$ and $K_\nu(X_t) \ge 0$.
So by Lemma~\ref{Lem:OUDer}, the first derivative of $H_\nu(X_t)$ is negative while the second derivative is positive, which means relative entropy is decreasing in a convex manner along the OU flow.

\subsection{Derivatives of mutual information along the OU flow}

Recall that given a functional $F(Y) \equiv F(\rho_Y)$ of a random variable $Y \sim \rho_Y$, we can define its {\em mutual version} $F(X;Y)$ for a joint random variable $(X,Y) \sim \rho_{XY}$ by 
\begin{align}\label{Eq:MutF}
F(X;Y) = F(Y\,|\,X) - F(Y)
\end{align}
where $F(Y\,|\,X) = \int \rho_X(x) F(\rho_{Y|X}(\cdot\,|\,x)) \, dx$ is the expected value of $F$ on the conditional random variables $Y|\{X=x\} \sim \rho_{Y|X}(\cdot\,|\,x)$, averaged over $X \sim \rho_X$.
The mutual version picks up only the nonlinear part, so two functionals that differ by a linear function have the same mutual version.

For example, for any $\nu$, the relative entropy $H_\nu(\rho)$ differs from the negative entropy $-H(\rho) = \int \rho(x) \log \rho(x) \, dx$ by the linear (in $\rho$) term $\int \rho(x) \log \nu(x) \, dx$.
Therefore, the mutual version of relative entropy is equal to the mutual version of negative entropy, which is {\em mutual information}:
$$H_\nu(X;Y) = I(X;Y) = H(Y) - H(Y\,|\,X).$$

We apply this definition to the joint random variable $(X,Y) = (X_0,X_t)$ where $X_t$ is the OU flow from $X_0$.
By the linearity of the OU channel, Lemma~\ref{Lem:OUDer} yields the following identities for the time derivatives of mutual information along the OU flow.

\begin{lemma}\label{Lem:OUMutDer}
Along the OU flow for $\nu = \N(0,\frac{1}{\alpha}I)$,
\begin{align*}
\frac{d}{dt} I(X_0;X_t) &= -J_\nu(X_0;X_t), \\
\frac{d^2}{dt^2} I(X_0;X_t) &= 2K_\nu(X_0;X_t) + 2\alpha J_\nu(X_0;X_t).
\end{align*}
\end{lemma}

We discuss the signs of the first and second derivatives.

\subsubsection{First derivative of mutual information}
\label{Sec:FirDerMut}

We recall that for a general joint random variable $(X,Y)$, the mutual relative Fisher information $J_\nu(X;Y) = J_\nu(Y\,|\,X) - J_\nu(Y)$ is equal to the {\em backward Fisher information} $\Phi(X\,|\,Y)$, which is manifestly nonnegative.
Furthermore, we also recall that along the OU flow, the backward Fisher information $\Phi(X_0\,|\,X_t)$ is proportional to the conditional variance of $X_0$ given $X_t$;
see for example~\cite[$\S$III-D.2]{WibisonoJL17}, or Appendix~\ref{App:mmse}.

\begin{lemma}\label{Lem:mmse}
Along the OU flow for $\nu = \N(0,\frac{1}{\alpha} I)$,
$$J_\nu(X_0;X_t) = \frac{\alpha^2 e^{-2\alpha t}}{(1-e^{-2\alpha t})^2} \Var(X_0\,|\,X_t).$$
\end{lemma}

Combining Lemma~\ref{Lem:mmse} with the first identity of Lemma~\ref{Lem:OUMutDer} expresses the time derivative of mutual information along the OU flow in terms of the minimum mean-square error (mmse) of estimating $X_0$ from $X_t$; this generalizes the I-MMSE relationship of~\cite{GuoEtAl05} from the Gaussian channel (the heat flow) to the OU flow.

\subsubsection{Second derivative of mutual information}

Unlike $J_\nu(X;Y)$ which is always positive, in general the mutual relative second-order Fisher information $K_\nu(X;Y) = K_\nu(Y\,|\,X) - K_\nu(Y)$ can be negative.
However, if $Y \sim \rho_Y$ is sufficiently log-concave compared to the reference measure $\nu$, then we can lower bound $K_\nu(X;Y)$ and use the additional term $\alpha J_\nu(X;Y)$ in the second identity of Lemma~\ref{Lem:OUMutDer} to offset it.
Altogether, we have the following result.

We say $X \sim \rho$ is $\lambda$-strongly log-concave for some $\lambda > 0$ if $-\nabla^2 \log \rho(x) \succeq \lambda I$ for all $x \in \R^n$.

\begin{lemma}\label{Lem:SuffConv}
Along the OU flow for $\nu = \N(0,\frac{1}{\alpha}I)$, if $X_t \sim \rho_t$ is $\frac{\alpha}{2}$-strongly log-concave, then mutual information $I(X_0;X_t)$ is convex at time $t$.
\end{lemma}

Lemma~\ref{Lem:SuffConv} provides a sufficient condition for mutual information to be convex along the OU flow.
Since the target measure $\nu = \N(0,\frac{1}{\alpha}I)$ is $\alpha$-strongly log-concave, we expect any initial distribution will eventually be transformed to a distribution that is at least $\frac{\alpha}{2}$-strongly log-concave.
We prove this in two cases: when the initial distribution is strongly log-concave, or when the initial distribution is bounded (or a convolution of the two cases).
First, we have the following classical estimate.

\begin{lemma}\label{Lem:SLC}
Along the OU flow for $\nu = \N(0,\frac{1}{\alpha}I)$,
if $X_0 \sim \rho_0$ is $\lambda$-strongly log-concave for some $\lambda > 0$,
then $X_t \sim \rho_t$ is $(\frac{e^{-2\alpha t}}{\lambda} + \frac{1-e^{-2\alpha t}}{\alpha})^{-1}$-strongly log-concave.
\end{lemma}

We say $X \sim \rho$ is $D$-bounded for some $D \ge 0$ if the support of $\rho$ is contained in a ball of diameter $D$.
Then we also have the following.

\begin{lemma}\label{Lem:Bdd}
Along the OU flow for $\nu = \N(0,\frac{1}{\alpha}I)$,
if $X_0 \sim \rho_0$ is $D$-bounded, then
$X_t \sim \rho_t$ is $\frac{\alpha(1-(1+D^2\alpha)e^{-2\alpha t})}{(1-e^{-2\alpha t})^2}$-strongly log-concave for $t \ge \frac{1}{2\alpha} \log (1+D^2 \alpha)$.
\end{lemma}

The threshold $t \ge \frac{1}{2\alpha} \log (1+D^2 \alpha)$ in Lemma~\ref{Lem:Bdd} is so that the log-concavity constant $\frac{\alpha(1-(1+D^2\alpha)e^{-2\alpha t})}{(1-e^{-2\alpha t})^2}$ is nonnegative.
It is possible to combine Lemmas~\ref{Lem:SLC} and~\ref{Lem:Bdd} to handle the case when the initial distribution is a convolution of a strongly log-concave and a bounded distribution (for example, a mixture of Gaussians), but with a more complicated threshold. For simplicity, we omit it here.

\section{Convexity of mutual information}

We present our main results on the convexity of mutual information along the OU flow.
Throughout, let $X_t \sim \rho_t$ denote the OU flow from $X_0 \sim \rho_0$.

\subsection{Eventual convexity when initial distribution is strongly log-concave}

By combining Lemmas~\ref{Lem:SuffConv} and~\ref{Lem:SLC}, we establish the following result for strongly log-concave distributions: 

\begin{theorem}\label{Thm:Conv}
Suppose $X_0 \sim \rho_0$ is $\lambda$-strongly log-concave for some $\lambda > 0$.
Along the OU flow for $\nu = \N(0,\frac{1}{\alpha}I)$, mutual information $t \mapsto I(X_0;X_t)$ is convex for all $t \ge t(\lambda)$, where
$$t(\lambda) =
\begin{cases}
0 ~~ & \text{ if } \lambda \ge \frac{\alpha}{2}, \\
\frac{1}{2\alpha} \log (\frac{\alpha}{\lambda}-1) ~ & \text{ if } \lambda < \frac{\alpha}{2}.
\end{cases}$$
\end{theorem}

Theorem~\ref{Thm:Conv} above proves the eventual convexity of mutual information for any strongly log-concave initial distribution.
However, the threshold is not tight.
For example, if $X_0 \sim \N(\mu,\Sigma)$, then $I(X_0;X_t) = \frac{1}{2} \sum_{i=1}^n \log (1+\frac{s_i\alpha e^{-2\alpha t}}{1-e^{-2\alpha t}})$ where $s_1 \ge \cdots \ge s_n > 0$ are the eigenvalues of $\Sigma$.
Then $X_0$ is $\frac{1}{s_1}$-strongly log-concave, but one can verify that in this case mutual information $I(X_0;X_t)$ is always convex for all $t \ge 0$, for any $s_1,\dots,s_n > 0$.
Ultimately this gap is due to the fact that the sufficient condition in Lemma~\ref{Lem:SuffConv} is not necessary, and it would be interesting to see how to tighten it.

\subsection{Eventual convexity when initial distribution is bounded}

By combining Lemmas~\ref{Lem:SuffConv} and~\ref{Lem:Bdd}, we establish the following result for distributions with bounded support.

\begin{theorem}\label{Thm:ConvBdd}
Suppose $X_0 \sim \rho_0$ is $D$-bounded for some $D \ge 0$.
Along the OU flow for $\nu = \N(0,\frac{1}{\alpha}I)$, mutual information $t \mapsto I(X_0;X_t)$ is convex for all $t \ge \frac{1}{2\alpha} \log (\sqrt{1+D^4\alpha^2}+D^2\alpha)$.
\end{theorem}

Note that as $\alpha \to 0$, the threshold in Theorem~\ref{Thm:ConvBdd} above becomes $t \ge \frac{D^2}{2}$, thus recovering the corresponding result for the heat flow from~\cite[Theorem~2]{WJ18}.
As noted previously, we can combine Theorems~\ref{Thm:Conv} and~\ref{Thm:ConvBdd} to show the eventual convexity of mutual information when the initial distribution is a convolution of a strongly log-concave and a bounded distribution, but for simplicity we omit it here.

\subsection{Eventual convexity when Fisher information is finite}

We now investigate the convexity of mutual information in general, regardless of the log-concavity of the distributions.
We can show that if the initial distribution has finite fourth moment and Fisher information, then mutual information is eventually convex along the OU flow.

For $p \ge 0$, let $M_p(X) = \E[\|X-\mu\|^p]$ denote the $p$-th moment of a random variable $X$ with mean $\E[X] = \mu$.
Let $J(X) = \int \rho(x) \|\nabla \log \rho(x)\|^2 dx$ denote the (absolute) Fisher information of $X \sim \rho$.
Then we have the following.

\begin{theorem}\label{Thm:EventConvFI}
Suppose $X_0 \sim \rho_0$ has $M \equiv M_4(X_0) < \infty$ and $J \equiv J(X_0) < \infty$.
Along the OU flow for $\nu = \N(0,\frac{1}{\alpha}I)$, mutual information $t \mapsto I(X_0;X_t)$ is convex for all $t \ge \frac{1}{2\alpha} \log(1 + \frac{\alpha MJ}{n^2} + \frac{1}{3}(1-\frac{\alpha MJ}{n^2}+\frac{\alpha^2 M}{n})^2)$.
\end{theorem}

Theorem~\ref{Thm:EventConvFI} above proves the eventual convexity of mutual information under rather general conditions.
However, in the limit $\alpha \to 0$ the time threshold becomes $+\infty$, rather than recovering the corresponding result for the heat flow from~\cite[Theorem~3]{WJ18}.
This is because in the proof we use a simple bound to estimate the root of a cubic polynomial (see Appendix~\ref{App:EventConvFI}), and it would be interesting to refine the analysis to obtain a better estimate.

\section{Nonconvexity of mutual information}

We now provide counterexamples to show that mutual information can be concave for small time along the OU flow.
Throughout, for $t > 0$, let $\tau_\alpha(t) = \frac{1}{\alpha}(1-e^{-2\alpha t})$.
For $u > 0$, let $V_u \in \R$ denote the one-dimensional Gaussian random variable $V_u \sim \N(u,u)$ with mean and variance equal to $u$.

\subsection{Mixture of two point masses}
\label{Sec:Pt2}

Let $X_0 \sim \frac{1}{2} \delta_{-\mu} + \frac{1}{2} \delta_\mu$ be a uniform mixture of two point masses centered at $-\mu$ and $\mu$, for some $\mu \in \R^n$, $\mu \neq 0$.
Along the OU flow for $\nu = \N(0,\frac{1}{\alpha}I)$, $X_t \sim \frac{1}{2} \N(-e^{-\alpha t} \mu, \tau_\alpha(t) I) + \frac{1}{2} \N(e^{-\alpha t} \mu, \tau_\alpha(t) I)$ is a uniform mixture of two Gaussians.

By direct calculation, the mutual information is:
$$I(X_0;X_t) = \frac{\alpha \|\mu\|^2}{e^{2\alpha t}-1} - \E\left[ \log \cosh \left(V_{\frac{\alpha \|\mu\|^2}{e^{2\alpha t}-1}}\right)\right].$$
The behavior is depicted in Figure~\ref{Fig:Pt} for $\alpha = \frac{1}{2}$ in $\R$.
We see that mutual information is not convex at small time.
It starts at the value $\log 2$, which follows from the general result in Theorem~\ref{Thm:GenMixt}, then stays flat for a while before decreasing and becoming convex.

\subsection{Mixture of two Gaussians}
\label{Sec:Gau2}

Let $X_0 \sim \frac{1}{2} \N(-\mu,sI) + \frac{1}{2} \N(\mu,sI)$ be a uniform mixture of two Gaussians with the same covariance $sI$ for some $s > 0$, centered at $-\mu$ and $\mu$ for some $\mu \in \R^n$, $\mu \neq 0$.
Along the OU flow for $\nu = \N(0,\frac{1}{\alpha}I)$, $X_t \sim \frac{1}{2} \N(-e^{-\alpha t} \mu, (e^{-2\alpha t}s +\tau_\alpha(t)) I) + \frac{1}{2} \N(e^{-\alpha t} \mu, (e^{-2\alpha t} s + \tau_\alpha(t)) I)$ is also a uniform mixture of two Gaussians.

By direct calculation, the mutual information is:
{\small
\begin{multline*}
I(X_0;X_t) = \frac{n}{2}\log\left(1+\frac{\alpha s}{e^{2\alpha t}-1}\right) \\
+ \frac{\alpha \|\mu\|^2}{\alpha s + e^{2\alpha t}-1} - \E\left[ \log \cosh \left(V_{\frac{\alpha \|\mu\|^2}{\alpha s + e^{2\alpha t}-1}}\right)\right].
\end{multline*}
}
The behavior is depicted in Figure~\ref{Fig:Gau} for $\alpha = \frac{1}{2}$ in $\R$ ($n=1$).
We see that now mutual information starts at $+\infty$, but it decreases quickly and flattens out for a while before decreasing again.
Therefore, mutual information is still not convex at some small time.

\begin{figure}[t!h!]
    \centering
    \begin{subfigure}[b]{0.2311\textwidth}
        \includegraphics[width=\textwidth]{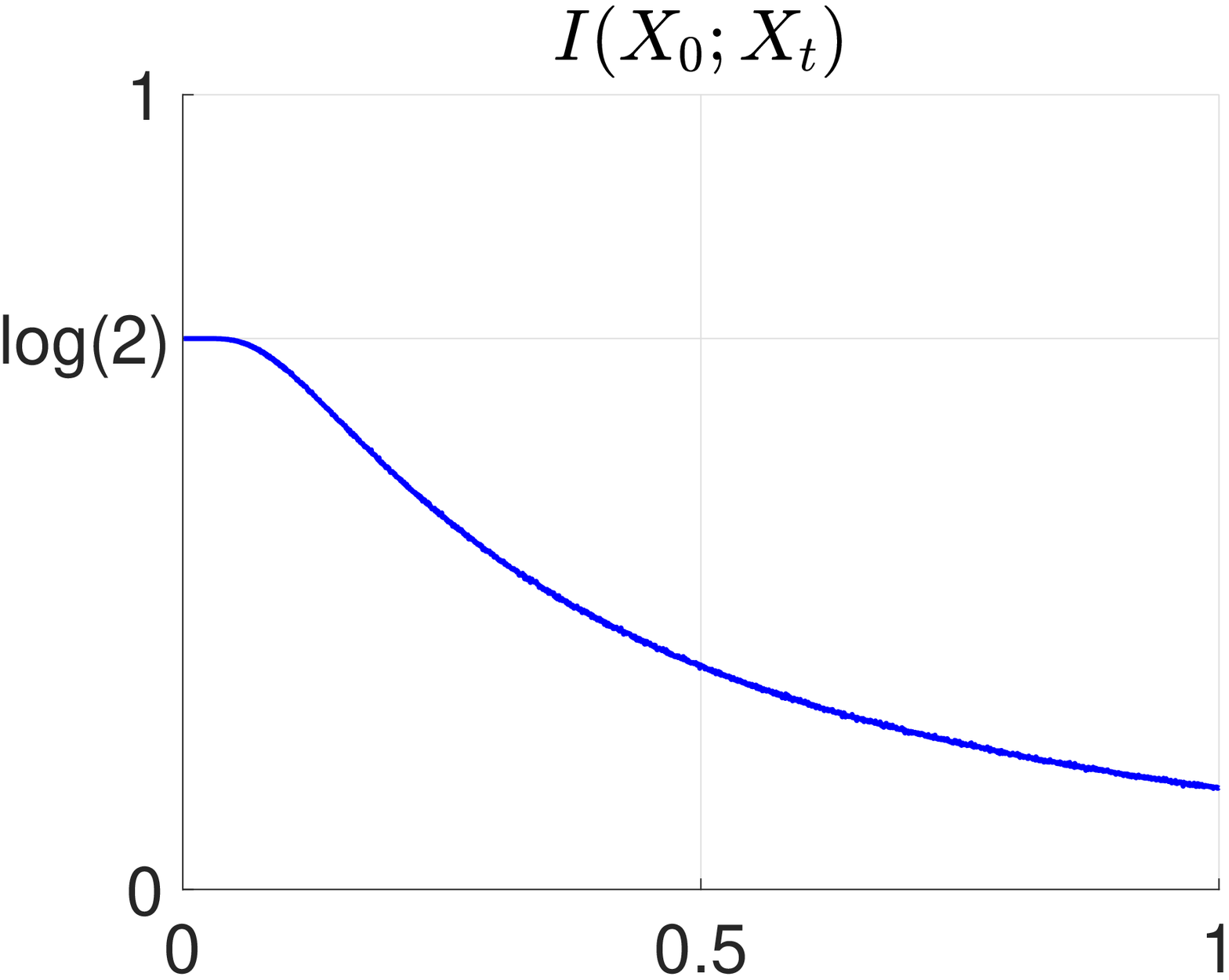}
        \caption{Mixture of point masses}
        \label{Fig:Pt}
    \end{subfigure}
      \; 
    \begin{subfigure}[b]{0.2311\textwidth}
        \includegraphics[width=\textwidth]{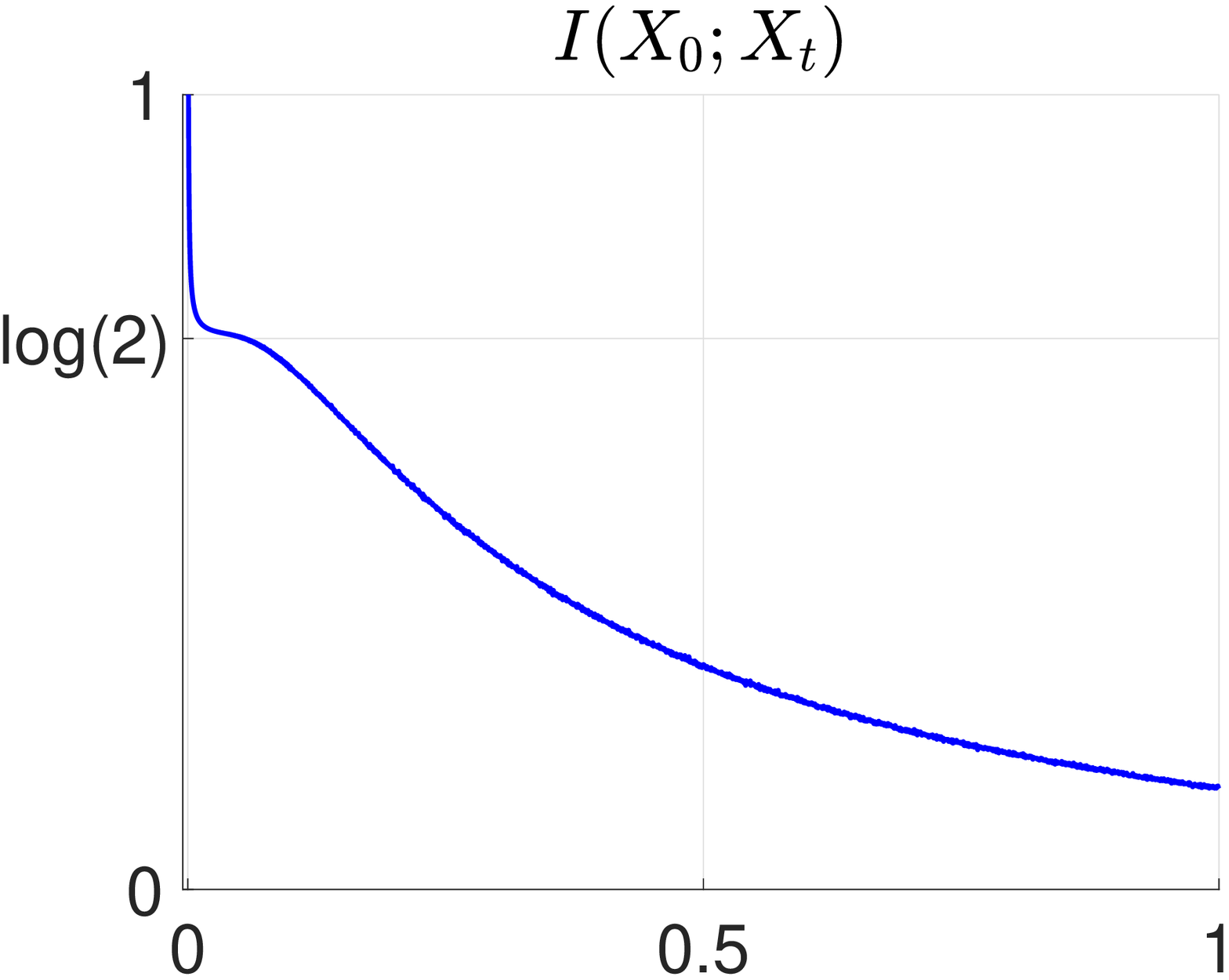}
        \caption{Mixture of Gaussians}
        \label{Fig:Gau}
    \end{subfigure} 
    
    \caption{Mutual information along the OU flow for $\nu = \N(0,2)$.
    (a) Left: $X_0 \sim \frac{1}{2} \delta_{-1} + \frac{1}{2} \delta_1$.
    (b) Right: $X_0 \sim \frac{1}{2} \N(-1,s) + \frac{1}{2} \N(1,s)$
     with $s = 10^{-3}$.}
    \label{Fig:MI}
\end{figure}

\subsection{General mixture of point masses}

Let $X_0 \sim \sum_{i=1}^k p_i \delta_{\mu_i}$ be a mixture of point masses centered at distinct $\mu_i \in \R^n$, with mixture probabilities $p_i > 0$, $\sum_{i=1}^k p_i = 1$.
Along the OU flow for $\nu = \N(0,\frac{1}{\alpha}I)$, $X_t \sim \sum_{i=1}^k p_i \N(e^{-\alpha t}\mu_i, \tau_\alpha(t) I)$ is a mixture of Gaussians.

By adapting the estimates from~\cite[$\S$IV-C]{WJ18}, we can show that mutual information along the OU flow starts at a finite value which is equal to the discrete entropy of the mixture probability, and it is exponentially concentrated at small time.
This is the same phenomenon as in the heat flow case, except that the bound on the small time now depends on $\alpha$.

Let $\|p\|_\infty = \max_{i,j} p_i/p_j$ and $m = \min_{i \neq j} \|\mu_i-\mu_j\| > 0$.
Let $h(p) = -\sum_{i=1}^k p_i \log p_i$ denote the discrete entropy.

\begin{theorem}\label{Thm:GenMixt}
Along the OU flow for $\nu = \N(0,\frac{1}{\alpha}I)$, 
for all $0 < t \le \frac{1}{2\alpha} \log(1+ \frac{\alpha^2 m^2}{676\|p\|_\infty^2})$,
$$0 \,\le\, h(p) - I(X_0;X_t) \,\le\, 3(k-1)\|p\|_\infty \exp\left(- \, \frac{0.085\,\alpha m^2}{e^{2\alpha t}-1}\right).$$
\end{theorem}

Theorem~\ref{Thm:GenMixt} above implies that $\lim_{t \to 0} I(X_0;X_t) = h(p).$
Thus, for discrete initial distribution, the initial value of mutual information only depends on the mixture proportions, and does not depend on the locations of the centers as long as they are distinct.
This is the same interesting behavior as in the heat flow case, and shows that we can obtain discontinuities of the mutual information at the origin with respect to the initial distribution, by moving the centers and merging them.

Furthermore, if a function converges exponentially fast, then all its derivatives converge to zero exponentially fast.
In our case for discrete initial distribution, this gives the following.

\begin{corollary}\label{Cor:Der}
For all $\ell \in \mathbb{N}$, $\lim_{t \to 0} \frac{d^\ell}{dt^\ell} I(X_0;X_t) = 0$.
\end{corollary}

In particular, the mutual relative Fisher information $J_\nu(X_0;X_t) = -\frac{d}{dt} I(X_0;X_t)$ also starts at $0$.
Since the initial distribution which is a mixture of point masses is bounded, by Theorem~\ref{Thm:ConvBdd} we know mutual information is eventually convex, which means $J_\nu(X_0;X_t)$ is eventually decreasing.
Since $J_\nu(X_0;X_t)$ is always nonnegative, it must initially increase first, before it can decrease.
During this period in which $J_\nu(X_0;X_t)$ is increasing, mutual information is concave.
This is similar to the behavior for the mixture of two point masses as observed in~$\S$\ref{Sec:Pt2}.
Moreover, by the continuity of the mutual relative first and second-order Fisher information with respect to the initial distribution, this suggests that mutual information can also be concave at some small time when the initial distribution is a mixture of Gaussians, similar to the observation in~$\S$\ref{Sec:Gau2}.

\section{Discussion and future work}

In this paper we have studied the convexity of mutual information along the Ornstein-Uhlenbeck flow. We considered the gradient flow interpretation of the Ornstein-Uhlenbeck process in the space of measures, and derived formulae for the various derivatives of relative entropy and mutual information.
We have shown that mutual information is eventually convex under rather general conditions on the initial distribution. We have also shown examples in which mutual information is concave at some small time.
These results generalize the behaviors seen in the heat flow~\cite{WJ18}.

For simplicity in this paper we have treated only the case when the target Gaussian distribution has isotropic covariance.
It is possible to extend our results to handle the case of a general covariance matrix. In this case, extra caution needs to be exercised since matrices in general do not commute. In the simple case when the covariance matrices of the initial and target distributions commute, our results extend naturally and the various thresholds are now controlled by the eigenvalues of the matrices.

As noted in the introduction, our interest in studying this problem is to better understand the general case of the Fokker-Planck process.
Indeed, there is an interesting dichotomy in which we understand the intricate properties of the Ornstein-Uhlenbeck process since we have an explicit solution, whereas we know very little about the general Fokker-Planck process.
The gradient flow interpretation applies to the general Fokker-Planck process and provides information for the convexity properties of the relative entropy if the target measure is log-concave.
However, much is not known, even about mutual information.
Hence in this paper we have attempted to settle the case of the Ornstein-Uhlenbeck process.
Even in this case some of our results are not tight and can be improved.

Some interesting future directions are to understand the convexity property of the solution to the Fokker-Planck process, even in the nice case when the target measure is strongly log-concave.
For example, does the Fokker-Planck process preserve log-concavity relative to the target measure?
Furthermore, is self-information (mutual information at initial time) for discrete initial distribution still equal to the discrete entropy for the Fokker-Planck process?
In general, it is interesting to bridge the gap in our understanding between the Ornstein-Uhlenbeck process and the general Fokker-Planck process.
One avenue to do that may be to study a perturbation of the Ornstein-Uhlenbeck process, when the target distribution is a small perturbation of the Gaussian.

\bibliographystyle{IEEEtran}
\bibliography{mi_ou_arxiv_v2.bbl}

\newpage

\appendix

\subsection{Some results for general Fokker-Planck flow}
\label{App:FP}

Let $\nu$ be a probability distribution in $\R^n$, which we represent via its density function with respect to the Lebesgue measure.
We assume $\nu \colon \R^n \to \R$ is smooth and positive everywhere, and let $f = -\log \nu$ be its negative log density.

We recall the {\em Fokker-Planck process} in $\R^n$ for the target measure $\nu$ is the stochastic differential equation
\begin{align}\label{Eq:FPProc}
dX = -\nabla f(X) \, dt + \sqrt{2} \, dW
\end{align}
where $X = (X_t)_{t \ge 0}$ is a stochastic process in $\R^n$ and $W = (W_t)_{t \ge 0}$ is the standard Brownian motion in $\R^n$.

In the space of measures, the stochastic process above corresponds to the {\em Fokker-Planck equation}, which is the following partial differential equation:
\begin{align}\label{Eq:FPEq}
\part{\rho}{t} = \nabla \cdot (\rho \nabla f) + \Delta \rho
\end{align}
where $\rho = \rho(x,t)$ for $x \in \R^n$, $t \ge 0$.
This means if the random variable $X_t \sim \rho_t$ evolves following the Fokker-Planck process~\eqref{Eq:FPProc}, then its probability density function $\rho(x,t) = \rho_t(x)$ evolves following the Fokker-Planck equation~\eqref{Eq:FPEq}.
We refer to the mapping $X_0 \mapsto X_t$ under~\eqref{Eq:FPProc}, or equivalently, the mapping $\rho_0 \mapsto \rho_t$ under~\eqref{Eq:FPEq}, as the {\em Fokker-Planck flow}.

We recall the interpretation of the Fokker-Planck flow as the gradient flow $\dot \rho = -\grad_\rho H_\nu$ of relative entropy
$$H_\nu(\rho) = \int_{\R^n} \rho(x) \log \frac{\rho(x)}{\nu(x)} \, dx$$
with respect to the target probability measure
$\nu = e^{-f}$
in the space of measures with the Wasserstein metric induced by the squared Euclidean metric in $\R^n$; see for example~\cite{JKO98,OV00,Vil03,Vil08}.
Then we can invoke general gradient flow identities to obtain information on the behavior of certain quantities---such as the relative entropy itself---along the Fokker-Planck flow.

For example, the derivative of the function value along its own gradient flow is given by the gradient squared.
That is, along $\dot \rho = -\grad_\rho H_\nu$, we have $\frac{d}{dt} H_\nu(\rho) = -\|\grad_\rho H_\nu\|^2_\rho$.
For us, this becomes the following identity on the derivative of relative entropy along the Fokker-Planck flow:
\begin{align}\label{Eq:FPDer}
\frac{d}{dt} H_\nu(\rho) = -J_\nu(\rho)
\end{align}
where
\begin{align}\label{Eq:RelJ}
J_\nu(\rho) = \int_{\R^n} \rho(x) \left\| \nabla \log \frac{\rho(x)}{\nu(x)} \right\|^2 \, dx
\end{align}
is the relative Fisher information of $\rho$ with respect to $\nu$.

Similarly, the second derivative of the function value along its own gradient flow is given by the Hessian operator applied to the gradient.
That is, along $\dot \rho = -\grad_\rho F$, we have $\frac{d^2}{dt^2} H_\nu(\rho) = 2 (\Hess_\rho H_\nu)(\grad_\rho H_\nu)$.
In our case, with the Hessian formula for relative entropy~\cite[Formula 15.7]{Vil08}, this becomes the following identity:
\begin{align}\label{Eq:FP2Der}
\frac{d^2}{dt^2} H_\nu(\rho) = 2K_\nu(\rho) + 2G_\nu^f(\rho)
\end{align}
where
\begin{align}\label{Eq:RelK}
K_\nu(\rho) = \int_{\R^n} \rho(x) \left\| \nabla^2 \log \frac{\rho(x)}{\nu(x)} \right\|^2_{\HS} dx
\end{align}
is the second-order relative Fisher information of $\rho$ with respect to $\nu$,
and
{\small
\begin{align}\label{Eq:Leftover}
G_\nu^f(\rho) = \int_{\R^n} \rho(x) \left\langle \nabla \log \frac{\rho(x)}{\nu(x)}, \big(\nabla^2 f(x)\big) \, \nabla \log \frac{\rho(x)}{\nu(x)} \right \rangle dx
\end{align}
}
is the leftover term.

\subsection{Proof of Lemma~\ref{Lem:OUDer}}
\label{App:OUDer}

The identities for the OU flow in Lemma~\ref{Lem:OUDer} follow from the general identities for the Fokker-Planck flow as described in Appendix~\ref{App:FP}, specialized to the case when $\nu = \N(0,\frac{1}{\alpha}I)$.

The first identity in Lemma~\ref{Lem:OUDer} is the same as the identity~\eqref{Eq:FPDer}.

The second identity in Lemma~\ref{Lem:OUDer} follows from the identity~\eqref{Eq:FP2Der}, together with the calculation that the leftover term~\eqref{Eq:Leftover} in the case $\nu = \N(0,\frac{1}{\alpha}I)$ with $\nabla^2 f(x) = \alpha I$ becomes:
\begin{align*}
G_\nu^f(\rho) &= \int_{\R^n} \rho(x) \left\langle \nabla \log \frac{\rho(x)}{\nu(x)}, \big(\alpha I\big) \, \nabla \log \frac{\rho(x)}{\nu(x)} \right \rangle \, dx \\
&= \alpha \int_{\R^n} \rho(x) \left\| \nabla \log \frac{\rho(x)}{\nu(x)} \right\|^2 \, dx \\
&= \alpha \, J_\nu(\rho)
\end{align*}
as desired.

\subsection{Proof of Lemma~\ref{Lem:OUMutDer}}
\label{App:OUMutDer}

The identities in Lemma~\ref{Lem:OUMutDer} follow from Lemma~\ref{Lem:OUDer} and the linearity of the OU channel.

Concretely, recall by Lemma~\ref{Lem:OUDer} that $\frac{d}{dt} H_\nu(X_t) = -J_\nu(X_t)$.
We apply this result to the conditional density $\rho_{X_t|X_0}(\cdot\,|\,x_0)$ to get $\frac{d}{dt} H_\nu(X_t\,|\,X_0\!=\!x_0) = -J_\nu(X_t\,|\,X_0\!=\!x_0)$ for each $x_0 \in \R^n$.
Taking expectation over $X_0 \sim \rho_0$ and interchanging the order of expectation and differentiation yields $\frac{d}{dt} H_\nu(X_t\,|\,X_0) = -J_\nu(X_t\,|\,X_0)$.
Combining this with the earlier result above yields
$\frac{d}{dt} I(X_0;X_t) = \frac{d}{dt} H_\nu(X_t) - \frac{d}{dt} H_\nu(X_t\,|\,X_0) = -J_\nu(X_t) + J_\nu(X_t\,|\,X_0) = -J_\nu(X_0;X_t)$, as desired.

The proof for the second identity in Lemma~\ref{Lem:OUMutDer} proceeds identically using the second identity in Lemma~\ref{Lem:OUDer}.

\subsection{Some results on general mutual relative Fisher information}
\label{App:Fish}

We review some results on mutual relative first and second-order Fisher information for general distributions.
These will be useful in proving Lemmas~\ref{Lem:mmse} and~\ref{Lem:SuffConv}.

Let $(X,Y) \sim \rho_{XY}$ be a joint random variable in $\R^n \times \R^n$ with a smooth density $\rho_{XY}$, which we can factorize into
$$\rho_{XY}(x,y) = \rho_Y(y) \, \rho_{X|Y}(x\,|\,y) = \rho_X(x) \, \rho_{Y|X}(y\,|\,x).$$

\subsubsection{First-order}

Recall that the {\em pointwise backward Fisher information matrix} of $X$ given $Y=y$ is
\begin{align}
&\widetilde \Phi(X\,|\,Y\!=\!y) \notag \\
&= \int_{\R^n} \rho_{X|Y}(x|y) (\nabla_y \log \rho_{X|Y}(x|y))(\nabla_y \log \rho_{X|Y}(x|y))^\top dx \notag \\
&= -\int_{\R^n} \rho_{X|Y}(x\,|\,y) \, \nabla^2_y \log \rho_{X|Y}(x\,|\,y) \, dx  \label{Eq:PhiPtw}
\end{align}
where the second equality follows from integration by parts, assuming the boundary terms vanish.
The {\em backward Fisher information matrix} of $X$ given $Y$ is the average of the pointwise matrices:
$$\widetilde \Phi(X\,|\,Y) = \int_{\R^n} \rho_Y(y) \, \widetilde \Phi(X\,|\,Y\!=\!y) \, dy.$$
The {\em backward Fisher information} of $X$ given $Y$ is the trace:
$$\Phi(X\,|\,Y) = \Tr(\widetilde \Phi(X\,|\,Y)).$$
Note that $\widetilde \Phi(X\,|\,Y\!=\!y) \succeq 0$ for all $y \in \R^n$, so $\widetilde \Phi(X\,|\,Y) \succeq 0$ and $\Phi(X\,|\,Y) \ge 0$.

Recall the relative Fisher information matrix of $Y \sim \rho_Y$ with respect to a reference distribution $\nu$ is
\begin{align}
&\widetilde J_\nu(Y) = \int_{\R^n} \rho_Y(y) \left(\nabla_y \log \frac{\rho_Y(y)}{\nu(y)}\right)\left(\nabla_y \log \frac{\rho_Y(y)}{\nu(y)}\right)^\top dy \notag \\
&~~~~~~~= -\int_{\R^n} \rho_Y(y) \, \nabla^2_y \log \frac{\rho_Y(y)}{\nu(y)} \,  dy \notag \\
&~~~~~~~~~~ + \int_{\R^n} \rho_Y(y) \, \left( \nabla^2 \nu(y) + \nabla \nu(y) \, \nabla \nu(y)^\top \right) \, dy \label{Eq:JTilMut}
\end{align}
where the second equality follows from integration by parts, assuming the boundary terms vanish.
Recall also the definition of the {\em mutual relative Fisher information matrix}: 
$$\widetilde J_\nu(X;Y) = \widetilde J_\nu(Y\,|\,X) - \widetilde J_\nu(Y)$$
and of its trace, the {\em mutual relative Fisher information}:
$$J_\nu(X;Y) = \Tr(\widetilde J_\nu(X;Y)) = J_\nu(Y\,|\,X) - J_\nu(Y).$$
In general, the mutual relative Fisher information is equal to the backward Fisher information.

\begin{lemma}\label{Lem:JPhi}
For any joint random variable $(X,Y)$ and any probability measure $\nu$,
$$J_\nu(X;Y) = \Phi(X\,|\,Y).$$
\end{lemma}
\begin{proof}
From the factorization
$$\rho_X(x) \, \rho_{Y|X}(y\,|\,x) = \rho_Y(y) \, \rho_{X|Y}(x\,|\,y)$$
we have
$$-\nabla^2_y \log \frac{\rho_{Y|X}(y\,|\,x)}{\nu(y)} = -\nabla^2_y \log \frac{\rho_Y(y)}{\nu(y)} - \nabla^2_y \log \rho_{X|Y}(x\,|\,y).$$
We integrate both sides with respect to $\rho_{XY}(x,y)$.
On the left-hand side, by first integrating over $\rho_{Y|X}(y\,|\,x)$ for each fixed $x \in \R^n$ and using the relation~\eqref{Eq:JTilMut}, we obtain
$$\widetilde J_\nu(Y\,|\,X) - \int_{\R^n} \rho_Y(y) \, \left( \nabla^2 \nu(y) + \nabla \nu(y) \, \nabla \nu(y)^\top \right) \, dy.$$
On the right-hand side, by the relations~\eqref{Eq:JTilMut} and~\eqref{Eq:PhiPtw} we obtain
\begin{align*}
&\widetilde J_\nu(Y) - \int_{\R^n} \rho_Y(y) ( \nabla^2 \nu(y) + \nabla \nu(y) \, \nabla \nu(y)^\top ) dy
+ \widetilde \Phi(X|Y).
\end{align*}
Combining the two lines above and canceling the common integral terms, we get
$\widetilde J_\nu(Y\,|\,X) = \widetilde J_\nu(Y) + \widetilde \Phi(X\,|\,Y)$.
Equivalently,
$\widetilde J_\nu(X;Y) = \widetilde J_\nu(Y\,|\,X) - \widetilde J_\nu(Y) = \widetilde \Phi(X\,|\,Y).$
Taking trace gives
$$J_\nu(X;Y) = \Tr(\widetilde J_\nu(X;Y)) = \Tr(\widetilde \Phi(X\,|\,Y)) = \Phi(X\,|\,Y)$$
as desired.
\end{proof}

\subsubsection{Second-order}

We now recall that the {\em pointwise backward second-order Fisher information} of $X$ given $Y=y$ is
$$\Psi(X\,|\,Y\!=\!y) = \int_{\R^n} \rho_{X|Y}(x\,|\,y) \, \|\nabla^2_y \log \rho_{X|Y}(x\,|\,y)\|^2_{\HS} \, dx.$$
The {\em backward second-order Fisher information} of $X$ given $Y$ is the average:
$$\Psi(X\,|\,Y) = \int_{\R^n} \rho_Y(y) \, \Psi(X\,|\,Y\!=\!y) \, dy.$$
Note that $\Psi(X\,|\,Y\!=\!y) \ge 0$ for all $y \in \R^n$, so $\Psi(X\,|\,Y) \ge 0$.
We also recall the definition of the mutual relative second-order Fisher information 
$K_\nu(X;Y) = K_\nu(Y\,|\,X) - K_\nu(Y)$.

\begin{lemma}\label{Lem:KPsi}
For any joint random variable $(X,Y)$ and any probability measure $\nu$,
\begin{multline*}
K_\nu(X;Y) = \Psi(X\,|\,Y) \, + \\ 2 \int_{\R^n} \rho_Y(y) \left\langle -\nabla^2 \log \frac{\rho_Y(y)}{\nu(y)}, \, \widetilde \Phi(X\,|\,Y\!=\!y) \right\rangle_{\HS} dy.
\end{multline*}
\end{lemma}
\begin{proof}
As before, we have the decomposition
$$-\nabla^2_y \log \frac{\rho_{Y|X}(y\,|\,x)}{\nu(y)} = -\nabla^2_y \log \frac{\rho_Y(y)}{\nu(y)} - \nabla^2_y \log \rho_{X|Y}(x\,|\,y).$$
Taking the squared norm on both sides and expanding, we get
\begin{align*}
&\left\|\nabla^2_y \log \frac{\rho_{Y|X}(y\,|\,x)}{\nu(y)}\right\|^2_{\HS} \\
&~~~~= \left\|\nabla^2_y \log \frac{\rho_Y(y)}{\nu(y)}\right\|^2_{\HS} + \|\nabla^2_y \log \rho_{X|Y}(x\,|\,y)\|^2_{\HS} \\ 
&~~~~~~~~ + 2 \left\langle -\nabla^2_y \log \frac{\rho_Y(y)}{\nu(y)}, -\nabla^2_y \log \rho_{X|Y}(x\,|\,y) \right\rangle_{\HS}.
\end{align*}
We integrate both sides with respect to $\rho_{XY}(x,y)$.
On the left-hand side we get $K_\nu(Y\,|\,X)$.
The first term on the right-hand side
gives $K_\nu(Y)$. 
The second term gives $\Psi(X\,|\,Y)$.
For the third term,
by first integrating over $\rho_{X|Y}(x\,|\,y)$ we obtain an inner product with $\widetilde \Phi(X\,|\,Y\!=\!y)$.
That is,
\begin{multline*}
K_\nu(Y\,|\,X) = K(Y) + \Psi_\nu(X\,|\,Y) \\
+ 2 \int_{\R^n} \rho_Y(y) \left\langle -\nabla^2_y \log \frac{\rho_Y(y)}{\nu(y)}, \, \widetilde \Phi(X\,|\,Y\!=\!y)\right\rangle_{\HS} dy.
\end{multline*}
This implies the desired expression for $K_\nu(X;Y) = K_\nu(Y\,|\,X)-K_\nu(Y)$.
\end{proof}

Specializing to the case of Gaussian target measure $\nu = \N(0,\frac{1}{\alpha}I)$, we have the following bound.

\begin{lemma}\label{Lem:KGau}
Let $\nu = \N(0,\frac{1}{\alpha}I)$.
For any joint random variable $(X,Y)$, if $Y \sim \rho_Y$ is $\lambda$-strongly log-concave for some $\lambda > 0$, then
$$K_\nu(X;Y) \ge 2(\lambda-\alpha) \Phi(X\,|\,Y).$$
\end{lemma}
\begin{proof}
Since $\nu = \N(0,\frac{1}{\alpha}I)$, we have $-\nabla^2 \nu(y) = \alpha I$, so the identity in Lemma~\ref{Lem:KPsi} becomes
\begin{multline*}
K_\nu(X;Y) = \Psi(X\,|\,Y) - 2\alpha \Phi(X\,|\,Y) \, + \\ 2 \int_{\R^n} \rho_Y(y) \left\langle -\nabla^2 \log \rho_Y(y), \, \widetilde \Phi(X\,|\,Y\!=\!y) \right\rangle_{\HS} dy.
\end{multline*}
Since $\rho_Y$ is $\lambda$-strongly log-concave, $-\nabla^2 \log \rho_Y(y) \succeq \lambda I$ for all $y \in \R^n$.
Since $\widetilde \Phi(X\,|\,Y\!=\!y) \succeq 0$, we can use this inequality in the inner product in the integral term above, to get
$$K_\nu(X;Y) \ge \Psi(X\,|\,Y) + 2(\lambda-\alpha) \Phi(X\,|\,Y).$$
Finally, since $\Psi(X\,|\,Y) \ge 0$, we can drop it to obtain the desired conclusion.
\end{proof}

\subsection{Proof of Lemma~\ref{Lem:mmse}}
\label{App:mmse}

Let $X_0 \sim \rho_0$ and let $X_t \sim \rho_0$ be the OU flow from $X_0$.
For $y \in \R^n$, let $\rho_{0|t}(\cdot\,|\,y)$ denote the conditional distribution of $X_0\,|\,\{X_t=y\}$.
For $t > 0$, let $\tau_\alpha(t) = \frac{1}{\alpha}(1-e^{-2\alpha t})$.
We begin with the following preliminary results.

\begin{lemma}\label{Lem:HessCov}
Along the OU flow for $\nu = \N(0,\frac{1}{\alpha}I)$, for all $x,y \in \R^n$,
\begin{align*} 
-\nabla^2_y \log \rho_{0|t}(x\,|\,y) = \frac{\alpha^2 e^{-2\alpha t}}{(1-e^{-2\alpha t})^2} \Cov(\rho_{0|t}(\cdot\,|\,y)).
\end{align*}
In particular, note that it is a constant in $x$.
\end{lemma}
\begin{proof}
From the explicit solution $X_t \stackrel{d}{=} e^{-\alpha t} X_0 + \sqrt{\tau_\alpha(t)} Z$, where $Z \sim \N(0,I)$ is independent of $X_0$, we can write
\begin{align*}
\rho_{0|t}(x\,|\,y) &\propto \rho_0(x) \cdot \rho_{t|0}(y\,|\,x) \\
&\propto \rho_0(x) \cdot e^{-\frac{\|y-e^{-\alpha t}x\|^2}{2 \tau_\alpha(t)}}  \\
&\propto \rho_0(x) \cdot e^{\frac{\alpha e^{-\alpha t}}{1-e^{-2\alpha t}} \langle y,x \rangle} \cdot e^{-\frac{\alpha^2 e^{-2\alpha t}}{2(1-e^{-2\alpha t})^2} \|x\|^2}
\end{align*}
where the proportionality above is in terms of $x$.
Therefore, we can write the conditional density $\rho_{0|t}(x\,|\,y)$ as an exponential family distribution over $x$ with parameter $\eta = \frac{\alpha e^{-\alpha t}}{1-e^{-2\alpha t}} y$:
$$\rho_{0|t}(x\,|\,y) = h(x) e^{\langle x,\eta \rangle - L(\eta)}$$
where $h(x) = \rho_0(x) e^{-\frac{\alpha^2 e^{-2\alpha t}}{2(1-e^{-2\alpha t})^2} \|x\|^2}$ is the base measure, and
$$L(\eta) = \log \int_{\R^n} h(x) e^{\langle x,\eta \rangle} \, dx$$
is the log-partition function, or normalizing constant.
Then by chain rule, we have
\begin{align*}
-\nabla^2_y \log \rho_{0|t}(x\,|\,y)
&= - \left(\part{\eta}{y}\right)  \left( \nabla^2_\eta \log \rho_{0|t}(x\,|\,y) \right) \left(\part{\eta}{y}\right)^\top \\
&= \frac{\alpha^2 e^{-2\alpha t}}{(1-e^{-2\alpha t})^2} \nabla^2_\eta L(\eta).
\end{align*}
By a general identity for exponential family~\cite{WJ08}, or simply by differentiating, we have that
$$\nabla^2_\eta L(\eta) = \Cov(\rho_{0|t}(\cdot\,|\,y)).$$
Combining the two expressions above yields the result.
\end{proof}

\begin{lemma}\label{Lem:PhiOU}
Along the OU flow for $\nu = \N(0,\frac{1}{\alpha}I)$,
$$\Phi(X_0\,|\,X_t) = \frac{\alpha^2 e^{-2\alpha t}}{(1-e^{-2\alpha t})^2} \Var(X_0\,|\,X_t).$$
\end{lemma}
\begin{proof}
For each $y \in \R^n$, by integrating the identity in Lemma~\ref{Lem:HessCov} with respect to $\rho_{0|t}(\cdot\,|\,y)$ and using the definition~\eqref{Eq:PhiPtw} on the left-hand side, we obtain
$$\widetilde \Phi(X_0\,|\,X_t=y) = \frac{\alpha^2 e^{-2\alpha t}}{(1-e^{-2\alpha t})^2} \Cov(\rho_{0|t}(\cdot\,|\,y)).$$
Now integrating with respect to $\rho_t(y)$ and using the definition of conditional covariance give
$$\widetilde \Phi(X_0\,|\,X_t) = \frac{\alpha^2 e^{-2\alpha t}}{(1-e^{-2\alpha t})^2} \Cov(X_0\,|\,X_t).$$
Finally, taking trace on both sides gives the desired conclusion.
\end{proof}

\begin{proof}[Proof of Lemma~\ref{Lem:mmse}]
Lemma~\ref{Lem:mmse} follows by combining Lemmas~\ref{Lem:JPhi} and~\ref{Lem:PhiOU}.
\end{proof}

\subsection{Proof of Lemma~\ref{Lem:SuffConv}}

By Lemma~\ref{Lem:KGau} for $(X_0,X_t)$ with $\lambda = \frac{\alpha}{2}$, and by using Lemma~\ref{Lem:JPhi}, we have
$$K_\nu(X_0;X_t) \ge -\alpha \Phi(X_0\,|\,X_t) = -\alpha J_\nu(X_0;X_t).$$
Then by the second identity in Lemma~\ref{Lem:OUMutDer}, we get
$$\frac{d^2}{dt^2} I(X_0;X_t) = 2 \big( K_\nu(X_0;X_t) + \alpha J_\nu(X_0;X_t)\big) \ge 0$$
which means mutual information is convex at time $t$.

\subsection{Proof of Lemma~\ref{Lem:SLC}}

We use the explicit solution $X_t = e^{-\alpha t} X_0 + \sqrt{\tau_\alpha(t)} Z$, where $\tau_\alpha(t) = \frac{1}{\alpha}(1-e^{-2\alpha t})$ and $Z \sim \N(0,I)$ is independent of $X_0$.

Since $X_0 \sim \rho_0$ is $\lambda$-strongly log-concave, $e^{-\alpha t} X_0$ is $(e^{2\alpha t} \lambda)$-strongly log-concave.
Furthermore, if $Z \sim \N(0,I)$, then $\sqrt{\tau_\alpha(t)} Z \sim \N(0,\tau_\alpha(t) I)$ is $\frac{1}{\tau_\alpha(t)}$-strongly log-concave.
Therefore, by a standard property of the preservation of strong log-concavity under convolution (e.g.,~\cite[Theorem~3.7(b)]{SW14}), $X_t = e^{-\alpha t} X_0 + \sqrt{\tau_\alpha(t)} Z$ is $(e^{-2\alpha t}\lambda^{-1} + \tau_\alpha(t))^{-1}$-strongly log-concave, as desired.

\subsection{Proof of Lemma~\ref{Lem:Bdd}}

We use the same notation as in Appendix~\ref{App:mmse}.
First, we have the following result.

\begin{lemma}\label{Lem:HesOU}
Along the OU flow for $\nu = \N(0,\frac{1}{\alpha}I)$, for all $y \in \R^n$,
\begin{align*} 
-\nabla^2_y \log \rho_{t}(y) = \frac{\alpha}{1-e^{-2\alpha t}}\Big(I - \frac{\alpha e^{-2\alpha t}}{1-e^{-2\alpha t}} \Cov(\rho_{0|t}(\cdot\,|\,y))\Big).
\end{align*}
\end{lemma}
\begin{proof}
From the factorization
$$\rho_t(y) \, \rho_{0|t}(x\,|\,y) = \rho_0(x) \, \rho_{t|0}(y\,|\,x)$$
and by the exact solution for the OU flow and by Lemma~\ref{Lem:HessCov},
\begin{align*}
-\nabla^2_y \log \rho_t(y) &= -\nabla^2_y \log \rho_{t|0}(y\,|\,x) + \nabla^2_y \log \rho_{0|t}(x\,|\,y) \\
&= \frac{\alpha}{1-e^{-2\alpha t}} I - \frac{\alpha^2 e^{-2\alpha t}}{(1-e^{-2\alpha t})^2} \Cov(\rho_{0|t}(\cdot\,|\,y)) \\
&= \frac{\alpha}{1-e^{-2\alpha t}}\Big(I - \frac{\alpha e^{-2\alpha t}}{1-e^{-2\alpha t}} \Cov(\rho_{0|t}(\cdot\,|\,y))\Big)
\end{align*}
as desired.
\end{proof}

Using Lemma~\ref{Lem:HesOU}, we can prove Lemma~\ref{Lem:Bdd}.

\begin{proof}[Proof of Lemma~\ref{Lem:Bdd}]
Since $X_0 \sim \rho_0$ is $D$-bounded by assumption,
the conditional distribution $X_0\,|\,\{X_t=y\} \sim \rho_{0|t}(\cdot\,|\,y)$ is also $D$-bounded, so
$\Cov(\rho_{0|t}(\cdot\,|\,y)) \preceq D^2 I$ for all $y \in \R^n$.
Then by Lemma~\ref{Lem:HesOU},
\begin{align*}
-\nabla^2_y \log \rho_{t}(y) 
&= \frac{\alpha}{1-e^{-2\alpha t}}\left(I - \frac{\alpha e^{-2\alpha t}}{1-e^{-2\alpha t}} \Cov(\rho_{0|t}(\cdot\,|\,y))\right) \\
&\succeq \frac{\alpha}{1-e^{-2\alpha t}}\left(1 - \frac{D^2 \alpha e^{-2\alpha t}}{1-e^{-2\alpha t}} \right) I \\
&= \frac{\alpha}{1-e^{-2\alpha t}}\left(\frac{1 - (1+D^2 \alpha) e^{-2\alpha t}}{1-e^{-2\alpha t}} \right) I.
\end{align*}
If $t \ge \frac{1}{2\alpha} \log(1+D^2 \alpha)$, then the last expression above is nonnegative.
Thus, we conclude that if $t \ge \frac{1}{2\alpha} \log(1+D^2 \alpha)$, then $X_t$ is $\frac{\alpha}{1-e^{-2\alpha t}}\left(\frac{1 - (1+D^2 \alpha) e^{-2\alpha t}}{1-e^{-2\alpha t}} \right)$-strongly log-concave.
\end{proof}

\subsection{Proof of Theorem~\ref{Thm:Conv}}

By Lemmas~\ref{Lem:SuffConv} and~\ref{Lem:SLC}, it suffices to determine when the strong log-concavity estimate $\lambda_t = (\frac{e^{-2\alpha t}}{\lambda} + \frac{1-e^{-2\alpha t}}{\alpha})^{-1}$ exceeds $\frac{\alpha}{2}$.
We consider the two cases.

\begin{itemize}
  \item If $\lambda \ge \frac{\alpha}{2}$, then $\lambda_t \ge (\frac{e^{-2\alpha t}}{\alpha/2} + \frac{1-e^{-2\alpha t}}{\alpha})^{-1} = \frac{\alpha}{1+e^{-2\alpha t}} \ge \frac{\alpha}{2}$ for all $t \ge 0$.
Thus, $t(\lambda) = 0$ in this case.

  \item If $\lambda < \frac{\alpha}{2}$, then the inequality $\lambda_t \ge \frac{\alpha}{2}$ is equivalent to $\frac{e^{-2\alpha t}}{\lambda} \le - \frac{1-e^{-2\alpha t}}{\alpha} + \frac{2}{\alpha} = \frac{1+e^{-2\alpha t}}{\alpha}$.
Solving for $t$ yields $e^{2\alpha t} \ge \frac{\alpha}{\lambda}-1$, or equivalently, $t \ge \frac{1}{2\alpha} \log(\frac{\alpha}{\lambda}-1)$.
Thus, $t(\lambda) = \frac{1}{2\alpha} \log(\frac{\alpha}{\lambda}-1)$ in this case.
Note that since $\lambda < \frac{\alpha}{2}$, $\frac{\alpha}{\lambda}-1 > 0$, so $t(\lambda) > 0$.
\end{itemize}

\subsection{Proof of Theorem~\ref{Thm:ConvBdd}}

By Lemmas~\ref{Lem:SuffConv} and~\ref{Lem:Bdd}, it suffices to determine when the strong log-concavity estimate $\lambda_t = \frac{\alpha(1-(1+D^2\alpha)e^{-2\alpha t})}{(1-e^{-2\alpha t})^2}$ exceeds $\frac{\alpha}{2}$.
Here we already assume $t \ge \frac{1}{2\alpha} \log(1+D^2 \alpha)$ so $\lambda_t \ge 0$.

Letting $w = e^{-2\alpha t}$ for simplicity, the inequality $\lambda_t \ge \frac{\alpha}{2}$ is equivalent to $\frac{\alpha(1-(1+D^2\alpha)w)}{(1-w)^2} \ge \frac{\alpha}{2}$.
Dividing both sides by $\alpha > 0$ and clearing the denominator, this is equivalent to $2(1-(1+D^2\alpha)w) \ge (1-w)^2$.
Expanding the square and simplifying, this is equivalent to $w^2 + 2D^2\alpha w - 1 \le 0$, or equivalently, $(w+D^2\alpha)^2 \le 1+D^4\alpha^2$.
Since $w = e^{-2\alpha t} > 0$, we can take square root on both sides to obtain $w \le \sqrt{1+D^4\alpha^2} - D^2\alpha$.
Finally, taking logarithm on both sides and using the relation $\frac{1}{\sqrt{1+D^4\alpha^2}-D^2 \alpha} = \sqrt{1+D^4\alpha^2}+D^2\alpha$, we conclude that the inequality $\lambda_t \ge \frac{\alpha}{2}$ above is equivalent to $t \ge \frac{1}{2\alpha} \log (\sqrt{1+D^4\alpha^2} + D^2\alpha)$.

\subsection{Proof of Theorem~\ref{Thm:EventConvFI}}
\label{App:EventConvFI}

We use the same notation as in Appendix~\ref{App:mmse}.
We first present the following preliminary results.

\begin{lemma}\label{Lem:KnuOU}
Along the OU flow for $\nu = \N(0,\frac{1}{\alpha}I)$,
\begin{multline*}
K_\nu(X_0;X_t) = \frac{2\alpha^3 e^{-4\alpha t}}{(1-e^{-2\alpha t})^3} \Var(X_0\,|\,X_t) \\
- \frac{\alpha^4 e^{-4\alpha t}}{(1-e^{-2\alpha t})^4} \E\left[\|\Cov(\rho_{0|t}(\cdot\,|\,X_t)\|^2_{\HS}\right].
\end{multline*}
\end{lemma}
\begin{proof}
By Lemma~\ref{Lem:HesOU} and since $-\nabla^2 \nu(y) = \alpha I$, we have
\begin{align*} 
-\nabla^2_y \log \frac{\rho_{t}(y)}{\nu(y)} = \frac{\alpha e^{-2\alpha t}}{1-e^{-2\alpha t}}\Big(I - \frac{\alpha}{1-e^{-2\alpha t}} \Cov(\rho_{0|t}(\cdot\,|\,y))\Big).
\end{align*}
Taking the squared norm and expanding, we get
\begin{align*}
\left\|-\nabla^2_y \log \frac{\rho_{t}(y)}{\nu(y)}\right\|^2_{\HS}
&= \frac{\alpha^2 e^{-4\alpha t}}{(1-e^{-2\alpha t})^2}\Big(\|I\|^2_{\HS} \\
& ~ -\frac{2\alpha}{1-e^{-2\alpha t}} \Var(\rho_{0|t}(\cdot\,|\,y)) \\
& ~ + \frac{\alpha^2}{(1-e^{-2\alpha t})^2} \|\Cov(\rho_{0|t}(\cdot\,|\,y))\|^2_{\HS}\Big).
\end{align*}
Using $\|I\|^2_{\HS} = n$ and integrating over $\rho_t(y)$, we get
\begin{align*}
K_\nu(X_t) &= \frac{n \alpha^2 e^{-4\alpha t}}{(1-e^{-2\alpha t})^2}
-\frac{2\alpha^3 e^{-4\alpha t}}{(1-e^{-2\alpha t})^3} \Var(X_0\,|\,X_t) \\
& ~ + \frac{\alpha^4 e^{-4\alpha t}}{(1-e^{-2\alpha t})^4} \E\left[\|\Cov(\rho_{0|t}(\cdot\,|\,X_t)\|^2_{\HS}\right].
\end{align*}
Finally, since $K_\nu(X_t\,|\,X_0) = \frac{n \alpha^2 e^{-4\alpha t}}{(1-e^{-2\alpha t})^2}$, the above gives the desired expression for $K_\nu(X_0;X_t) = K_\nu(X_t\,|\,X_0) - K_\nu(X_t)$.
\end{proof}

Recall that $J(X) = \int_{\R^n} \rho(x) \|\nabla \log \rho(x)\|^2 dx$ is the (absolute) Fisher information of $X \sim \rho$.

\begin{lemma}\label{Lem:JX0XtOU}
Along the OU flow for $\nu = \N(0,\frac{1}{\alpha}I)$,
$$J(X_0\,|\,X_t) = J(X_0) + \frac{n\alpha e^{-2\alpha t}}{1-e^{-2\alpha t}}.$$
\end{lemma}
\begin{proof}
From the factorization
$$\rho_t(y) \, \rho_{0|t}(x\,|\,y) = \rho_0(x) \, \rho_{t|0}(y\,|\,x)$$
and by the exact solution for the OU flow we have
\begin{align*}
-\nabla_x \log \rho_{0|t}&(x\,|\,y) = -\nabla_x \log \rho_0(x) - \nabla_x \log \rho_{t|0}(y\,|\,x) \\
&~~~~~~ = -\nabla_x \log \rho_0(x) + \frac{\alpha e^{-\alpha t}}{1-e^{-2\alpha t}}(e^{-\alpha t}x-y).
\end{align*}
Squaring both sides and expanding, we get
\begin{align*}
\|-\nabla_x \log \rho_{0|t}&(x\,|\,y)\|^2 = \|-\nabla_x \log \rho_0(x)\|^2 \\
&+ \frac{\alpha^2 e^{-2\alpha t}}{(1-e^{-2\alpha t})^2}\|e^{-\alpha t}x-y\|^2 \\
&+ \frac{2\alpha e^{-\alpha t}}{1-e^{-2\alpha t}} \langle -\nabla_x \log \rho_0(x), \, e^{-\alpha t}x-y \rangle.
\end{align*}
We integrate both sides with respect to $\rho_{0t}(x,y)$.
On the left-hand side we get $J(X_0\,|\,X_t)$.
On the right-hand side, the first term gives $J(X_0)$.
Since $\rho_{t|0}(\cdot\,|\,x) = \N(e^{-\alpha t}x, \frac{1-e^{-2\alpha t}}{\alpha} I)$, the second term gives $\frac{\alpha^2 e^{-2\alpha t}}{(1-e^{-2\alpha t})^2} \cdot \frac{(1-e^{-2\alpha t})n}{\alpha} = \frac{n\alpha e^{-2\alpha t}}{1-e^{-2\alpha t}}$.
And by integrating over $\rho_{t|0}(y\,|\,x)$ first for each fixed $x \in \R^n$, the third term gives $0$.
Therefore,
$J(X_0\,|\,X_t) = J(X_0) + \frac{n\alpha e^{-2\alpha t}}{1-e^{-2\alpha t}}$,
as desired.
\end{proof}

By Lemma~\ref{Lem:JX0XtOU}, we have the following lower bound on the conditional variance along the OU flow.

\begin{lemma}\label{Lem:VarOU}
Along the OU flow for $\nu = \N(0,\frac{1}{\alpha}I)$,
$$\Var(X_0\,|\,X_t) \ge \frac{n^2}{J(X_0) + \frac{n\alpha e^{-2\alpha t}}{1-e^{-2\alpha t}}}.$$
\end{lemma}
\begin{proof}
For any random variable $X \sim \rho$ in $\R^n$ with a smooth density, we recall the variance-Fisher information comparison inequality (e.g., see~\cite{Dembo91}):
$$\Var(X) \, J(X) \ge n^2$$
which also follows directly from Cauchy-Schwarz inequality and integration by parts.
We apply this to the conditional density $\rho_{0|t}(\cdot\,|\,y)$ for each $y \in \R^n$ to get
$$\Var(\rho_{0|t}(\cdot\,|\,y)) \ge \frac{n^2}{J(\rho_{0|t}(\cdot\,|\,y))}.$$
We take expectation on both sides over $Y = X_t$, and use the estimate $\E[\frac{1}{J}] \ge \frac{1}{\E[J]}$ by Cauchy-Schwarz, to get
$$\Var(X_0\,|\,X_t) \ge \E\left[\frac{n^2}{J(\rho_{0|t}(\cdot\,|\,X_t))}\right] \ge \frac{n^2}{J(X_0\,|\,X_t)}.$$
Finally, we plug in the expression for $J(X_0\,|\,X_t)$ from Lemma~\ref{Lem:JX0XtOU} to get the desired result.
\end{proof}

Recall that $M_4(X) = \E[\|X-\mu\|^4]$ is the fourth moment of a random variable $X$ with mean $\E[X] = \mu$.
We have the following general estimate; see also~\cite[Lemma~14]{WJ18}.

\begin{lemma}\label{Lem:M4}
For any joint random variable $(X,Y)$,
$$\E[\|\Cov(\rho_{X|Y}(\cdot \,|\,Y))\|^2_{\HS}] \le M_4(X).$$
\end{lemma}
\begin{proof}
Let $\mu = \E[X]$.
For each $y \in \R^n$,
\begin{align*}
\|\Cov(\rho_{X|Y}(\cdot\,|\,y))\|^2_{\HS} &\le (\Tr(\Cov(\rho_{X|Y}(\cdot\,|\,y))))^2 \\
&= (\Var(\rho_{X|Y}(\cdot\,|\,y)))^2 \\
&\le \left(\E_{\rho_{X|Y}(\cdot\,|\,y)}[\|X-\mu\|^2]\right)^2 \\
&\le \E_{\rho_{X|Y}(\cdot\,|\,y)}[\|X-\mu\|^4].
\end{align*}
The first inequality above follows from $\sum_{i=1}^n \lambda_i(y)^2 \le (\sum_{i=1}^n \lambda_i(y))^2$, where $\lambda_i(y) \ge 0$ are the eigenvalues of $\Cov(\rho_{X|Y}(\cdot\,|\,y))$;
the second inequality follows from the definition of the variance as the minimum square deviation around any point;
and the third inequality follows from Cauchy-Schwarz.
Now taking expectation over $Y$ and using the tower property of expectation, we obtain the desired result.
\end{proof}

We also have the following simple estimate on the root of a cubic polynomial.

\begin{lemma}\label{Lem:Cubic}
Consider the cubic polynomial $p(w) = w^3 + aw^2 + bw + 1$, $w \in \R$, for some $a,b \in \R$.
Note $p(0) = 1 > 0$, and assume $p(1) = 2+a+b < 0$, so there exists a unique root $w_0 \in (0,1)$.
Then we have $w_0 \ge 1/(\frac{a^2}{3}-b)$.
\end{lemma}
\begin{proof}
At each point $w$, the tangent line to the cubic polynomial has slope $3w^2+2aw+b$, which is bounded below by $-\frac{a^2}{3}+b$.
Starting at the point $(w,p(w)) = (0,1)$, the line with slope $-\frac{a^2}{3}+b$ crosses the $w$-axis at $1/(\frac{a^2}{3}-b)$.
Therefore, the root $w_0$ must be at least $1/(\frac{a^2}{3}-b)$, as desired.
\end{proof}

We are now ready to prove Theorem~\ref{Thm:EventConvFI}.

\begin{proof}[Proof of Theorem~\ref{Thm:EventConvFI}]
By the second identity in Lemma~\ref{Lem:OUMutDer}, and using the formulae in Lemmas~\ref{Lem:mmse} and~\ref{Lem:KnuOU},
followed by the inequalities in Lemmas~\ref{Lem:VarOU} and~\ref{Lem:M4}, we have
\begin{align*}
\frac{1}{2} & \frac{d^2}{dt^2} I(X_0;X_t) = K_\nu(X_0;X_t) + \alpha J_\nu(X_0;X_t) \\
&= \left(\frac{2\alpha^3 e^{-4\alpha t}}{(1-e^{-2\alpha t})^3} + \frac{\alpha \cdot \alpha^2 e^{-2\alpha t}}{(1-e^{-2\alpha t})^2} \right) \Var(X_0\,|\,X_t) \\
& ~~~~~~ - \frac{\alpha^4 e^{-4\alpha t}}{(1-e^{-2\alpha t})^4} \E\left[\|\Cov(\rho_{0|t}(\cdot\,|\,X_t)\|^2_{\HS}\right] \\
&= \frac{\alpha^3 e^{-2\alpha t}}{(1-e^{-2\alpha t})^4} \Big( (1-e^{-4\alpha t}) \Var(X_0\,|\,X_t) \\
& ~~~~~~~~~~~~~ - \alpha e^{-2\alpha t} \E\left[\|\Cov(\rho_{0|t}(\cdot\,|\,X_t)\|^2_{\HS}\right] \Big) \\
&\ge \frac{\alpha^3 e^{-2\alpha t}}{(1-e^{-2\alpha t})^4} \Big( (1-e^{-4\alpha t}) \cdot \frac{n^2}{(J(X_0) + \frac{n\alpha e^{-2\alpha t}}{1-e^{-2\alpha t}})} \\
& ~~~~~~~~~~~~~ - \alpha e^{-2\alpha t} \cdot M_4(X_0) \Big).
\end{align*}
To show $I(X_0;X_t)$ is convex at $t$, it suffices to find when the last expression in the parenthesis above is nonnegative.

For simplicity let $w = e^{-2\alpha t} \in (0,1)$, $J \equiv J(X_0)$, and $M \equiv M_4(X_0)$.
Then mutual information is convex when
$$\frac{n^2(1-w^2)}{J + \frac{n\alpha w}{1-w}} \ge \alpha w M.$$
Clearing the denominator and simplifying, this is equivalent to $n^2(1-w^2)(1-w) \ge \alpha w(1-w) MJ + n \alpha^2 w^2 M$.
Expanding and simplifying further, we get that mutual information is convex when the cubic polynomial
$$p(w) = w^3 + \Big(-1+\frac{\alpha MJ}{n^2} - \frac{\alpha^2 M}{n}\Big) w^2 + \Big(-1-\frac{\alpha MJ}{n^2}\Big) w + 1$$
is nonnegative.
Observe that $p(0) = 1 > 0$ and $p(1) = -\frac{\alpha^2 M}{n} < 0$.
Moreover, $\lim_{w \to +\infty} p(w) = +\infty$ and $\lim_{w \to -\infty} p(w) = -\infty$.
Therefore, $p$ has a unique root $w_0 \in (0,1)$, and $p(w) \ge 0$ for $0 < w \le w_0$.
By Lemma~\ref{Lem:Cubic} with $a = -1+\frac{\alpha MJ}{n^2} - \frac{\alpha^2 M}{n}$ and $b = -1-\frac{\alpha MJ}{n^2}$, we have the estimate $w_0 \ge 1/(\frac{a^2}{3}-b)$.
Therefore, $p(w) \ge 0$ when $w \le 1/(\frac{a^2}{3}-b)$.
Since $w = e^{-2\alpha t}$, the latter is equivalent to
$$t \ge \frac{1}{2\alpha} \log \left(\frac{a^2}{3}-b\right).$$
Plugging in the definitions of $a$ and $b$ gives us the desired conclusion.
\end{proof}

\subsection{Proof of Theorem~\ref{Thm:GenMixt}}

We first recall the following estimate from~\cite{WJ18}.

\begin{lemma}[{from~\cite[Lemma~15]{WJ18}}]\label{Lem:Log2}
Let $b > 0$, $c \ge \max\{1,\frac{26}{b}\}$, and $Z \sim \N(0,1)$.
Then
$$\E[\log(1 + be^{cZ-\frac{c^2}{2}})\big] \le 3be^{-0.085 c^2}.$$
\end{lemma}

We can prove Theorem~\ref{Thm:GenMixt} by adapting the calculations from~\cite[Theorem~4]{WJ18}.

\begin{proof}[Proof of Theorem~\ref{Thm:GenMixt}]
Since $X_0 \sim \sum_{i=1}^k p_i \delta_{\mu_i}$, along the OU flow for $\nu = \N(0,\frac{1}{\alpha}I)$, $X_t \sim \sum_{i=1}^k p_i \N(e^{-\alpha t} \mu_i, \tau_\alpha(t) I)$ where $\tau_\alpha(t) = \frac{1}{\alpha}(1-e^{-2\alpha t})$.
The density $\rho_t$ of $X_t$ is 
$$\rho_t(y) = \frac{1}{(2\pi \tau_\alpha(t))^{n/2}} \sum_{i=1}^k p_i e^{-\frac{\|y-e^{-\alpha t} \mu_i\|^2}{2\tau_\alpha(t)}}.$$
The entropy of $X_t$ is
$$H(X_t)
= \frac{n}{2} \log (2\pi \tau_\alpha(t)) - \E\left[ \log \left( \sum_{i=1}^k p_i e^{-\frac{\|X_t-e^{-\alpha t} \mu_i\|^2}{2\tau_\alpha(t)}} \right) \right].$$
The expectation is over the mixture $X_t \sim \sum_{i=1}^k p_i \N(e^{-\alpha t} \mu_i, \tau_\alpha(t) I)$, which we split into a sum over $i=1,\dots,k$ of the individual expectations over $Y \sim \N(e^{-\alpha t} \mu_i, \tau_\alpha(t) I)$.
When $Y \sim \N(e^{-\alpha t} \mu_i, \tau_\alpha(t) I)$, we write $Y = e^{-\alpha t}\mu_i + \sqrt{\tau_\alpha(t)} Z$ where $Z \sim \N(0,I)$.
Then we can write the entropy above as
{\small
\begin{align*}
&H(X_t) - \frac{n}{2} \log (2\pi \tau_\alpha(t)) \\
=& - \sum_{i=1}^k p_i \E\Big[\log\Big(p_i e^{-\frac{\|Z\|^2}{2}} + \sum_{j \neq i} p_j e^{-\frac{\|\sqrt{\tau_\alpha(t)}Z+e^{-\alpha t}(\mu_i-\mu_j)\|^2}{2\tau_\alpha(t)}}\Big)\Big]
\end{align*}
\begin{align*}
&= - \sum_{i=1}^k p_i \E\Big[ \log p_i - \frac{\|Z\|^2}{2} \\
&~~~~~ + \log \Big(1 +  \sum_{j \neq i} \frac{p_j}{p_i} e^{\frac{\|Z\|^2}{2}-\frac{\|\sqrt{\tau_\alpha(t)}Z+e^{-\alpha t}(\mu_i-\mu_j)\|^2}{2\tau_\alpha(t)}}\Big)\Big] \\
&= h(p) + \frac{n}{2} \\
&~~~~~ - \sum_{i=1}^k p_i \E\Big[ \log \Big(1 +  \sum_{j \neq i} \frac{p_j}{p_i} e^{\frac{\|Z\|^2}{2}-\frac{\|\sqrt{\tau_\alpha(t)}Z+e^{-\alpha t}(\mu_i-\mu_j)\|^2}{2\tau_\alpha(t)}}\Big)\Big]
\end{align*}
}
where $h(p) = -\sum_{i=1}^k p_i \log p_i$ is the discrete entropy.

Since $H(X_t\,|\,X_0) = \frac{n}{2} \log (2\pi e \tau_\alpha(t))$,
we have for mutual information
\begin{multline*}
h(p) - I(X_0;X_t) \\
= \sum_{i=1}^k p_i \E\Big[ \log \Big(1 +  \sum_{j \neq i} \frac{p_j}{p_i} e^{\frac{\|Z\|^2}{2}-\frac{\|\sqrt{\tau_\alpha(t)}Z+e^{-\alpha t}(\mu_i-\mu_j)\|^2}{2\tau_\alpha(t)}}\Big)\Big].
\end{multline*}

Clearly $h(p)-I(X_0;X_t) \ge 0$ since the logarithm on the right-hand side above is positive.

On the other hand, using the inequality $\log(1+\sum_j x_j) \le \sum_j \log(1+x_j)$ for $x_j > 0$, we also have the upper bound
\begin{multline*}
h(p)-I(X_0;X_t) \\
\le \sum_{i=1}^k p_i \sum_{j \neq i} \E\Big[ \log \Big(1 + \frac{p_j}{p_i} e^{\frac{\|Z\|^2}{2}-\frac{\|\sqrt{\tau_\alpha(t)}Z+e^{-\alpha t}(\mu_i-\mu_j)\|^2}{2\tau_\alpha(t)}}\Big)\Big].
\end{multline*}
For each $i \neq j$, the exponent on the right-hand side above is
\begin{multline*}
\frac{\|Z\|^2}{2}-\frac{\|\sqrt{\tau_\alpha(t)}Z+e^{-\alpha t}(\mu_i-\mu_j)\|^2}{2\tau_\alpha(t)} \\
= -\frac{e^{-\alpha t}}{\sqrt{\tau_\alpha(t)}} \langle Z, \mu_i-\mu_j \rangle - \frac{e^{-2\alpha t}\|\mu_i-\mu_j\|^2}{2\tau_\alpha(t)}
\end{multline*}
which has the $\N(-\frac{e^{-2\alpha t}\|\mu_i-\mu_j\|^2}{2\tau_\alpha(t)},\frac{e^{-2\alpha t}\|\mu_i-\mu_j\|^2}{\tau_\alpha(t)})$ distribution in $\R$, so it has the same distribution as $-\frac{e^{-2\alpha t}\|\mu_i-\mu_j\|^2}{2\tau_\alpha(t)} + \frac{e^{-\alpha t}\|\mu_i-\mu_j\|}{\sqrt{\tau_\alpha(t)}} Z_1$ where $Z_1 \sim \N(0,1)$ is the standard one-dimensional Gaussian.
Thus, we can write the upper bound above as
$$h(p) - I(X;Y) \le \sum_{i=1}^k p_i \sum_{j \neq i} \E\left[ \log \left(1 + b_{ij} e^{c_{ij} Z_1 - \frac{c_{ij}^2}{2}}\right)\right]$$
where $b_{ij} = \frac{p_j}{p_i}$, $c_{ij} = \frac{e^{-\alpha t}\|\mu_i-\mu_j\|}{\sqrt{\tau_\alpha(t)}}$, and $Z_1 \sim \N(0,1)$ in $\R$.

By Lemma~\ref{Lem:Log2}, if $c_{ij} \ge \max\{1,\frac{26}{b_{ij}}\}$, then we have
\begin{align*}
h(p) - I(X_0;X_t) \le 3 \sum_{i=1}^k p_i \sum_{j \neq i} b_{ij} e^{-0.085c_{ij}^2}.
\end{align*}
Note that $b_{ij} = \frac{p_j}{p_i} \le \|p\|_\infty$ and $c_{ij}^2 = \frac{e^{-2\alpha t}\|\mu_i-\mu_j\|^2}{\tau_\alpha(t)} \ge \frac{e^{-2\alpha t} m^2}{\tau_\alpha(t)} = \frac{\alpha m^2}{e^{2\alpha t}-1}$, so
\begin{align*}
h(p) - I(X_0;X_t) &\le 3 \sum_{i=1}^k p_i \sum_{j \neq i} \|p\|_\infty e^{-0.085 \frac{\alpha m^2}{e^{2\alpha t}-1}} \\
&= 3(k-1) \|p\|_\infty e^{-0.085 \frac{\alpha m^2}{e^{2\alpha t}-1}}.
\end{align*}
Now, the condition $c_{ij} \ge \max\{1,\frac{26}{b_{ij}}\}$ is equivalent to 
$$e^{2\alpha t} \le 1+\frac{\alpha^2 \|\mu_i-\mu_j\|^2}{\max\{1,\frac{26}{b_{ij}}\}^2}.$$
Since $\|\mu_i-\mu_j\|^2 \ge m^2$ and $\frac{1}{b_{ij}} = \frac{p_i}{p_j} \le \|p\|_\infty$,
 the condition above is satisfied when
$$e^{2\alpha t} \le 1+ \frac{\alpha^2 m^2}{\max\{1,26\|p\|_\infty\}^2} = 1+\frac{\alpha^2 m^2}{676\|p\|_\infty^2}.$$

Thus, we conclude that if $t \le \frac{1}{2\alpha} \log(1+ \frac{\alpha^2 m^2}{676\|p\|_\infty^2})$, then
\begin{align*}
h(p) - I(X_0;X_t) \le 3 (k-1) \|p\|_\infty e^{-0.085\frac{\alpha m^2}{e^{2\alpha t}-1}}
\end{align*}
as desired.
\end{proof}

\subsection{Proof of Corollary~\ref{Cor:Der}}

From Theorem~\ref{Thm:GenMixt}, we have for small $t$,
$$\left|\frac{I(X_0;X_t) - h(p)}{t^\ell} \right| \le 3(k-1) \|p\|_\infty \frac{e^{-0.085 \frac{\alpha m^2}{e^{2\alpha t}-1}}}{t^\ell}.$$
The right-hand side tends to $0$ as $t \to 0$.
Inductively, this implies all derivatives of $I(X_0;X_t)$ converge to $0$ exponentially fast as $t \to 0$.

\end{document}